\documentclass[runningheads]{llncs}
\usepackage[T1]{fontenc}
\usepackage{graphicx}
\usepackage{amsmath}
\usepackage{amssymb}
\usepackage{amsfonts}
\usepackage{bm}

\usepackage{cases}
\usepackage{tikz}
\usepackage{mathtools}
\usepackage{hyperref}
\newtheorem{assumption}{Assumption}

\begin{document}
\title{Relationship between H\"{o}lder Divergence and Functional Density Power Divergence: Intersection and Generalization}
\titlerunning{Relationship between H\"{o}lder Divergence and Functional Density Power}
%
\author{Masahiro Kobayashi\orcidID{0000-0002-0278-6153}}
\authorrunning{M. Kobayashi}
%
\institute{Information and Media Center, Toyohashi University of Technology, 1-1 Hibarigaoka, Tempaku-cho, Toyohashi, Aichi, 441-8580, Japan\\ \email{kobayashi@imc.tut.ac.jp}}
\maketitle              
\begin{abstract}
In this study, we discuss the relationship between two families of density-power-based divergences with functional degrees of freedom---the H\"{o}lder divergence and the functional density power divergence (FDPD)---based on their intersection and generalization.
These divergence families include the density power divergence and the $\gamma$-divergence as special cases.
First, we prove that the intersection of the H\"{o}lder divergence and the FDPD is limited to a general divergence family introduced by Jones et al. (Biometrika, 2001).
Subsequently, motivated by the fact that H\"{o}lder's inequality is used in the proofs of nonnegativity for both the H\"{o}lder divergence and the FDPD, we define a generalized divergence family, referred to as the $\xi$-H\"{o}lder divergence.
The nonnegativity of the $\xi$-H\"{o}lder divergence is established through a combination of the inequalities used to prove the nonnegativity of the H\"{o}lder divergence and the FDPD.
Furthermore, we derive an inequality between the composite scoring rules corresponding to different FDPDs based on the $\xi$-H\"{o}lder divergence.
Finally, we prove that imposing the mathematical structure of the H\"{o}lder score on a composite scoring rule results in the $\xi$-H\"{o}lder divergence.

\end{abstract}
\section{Introduction}
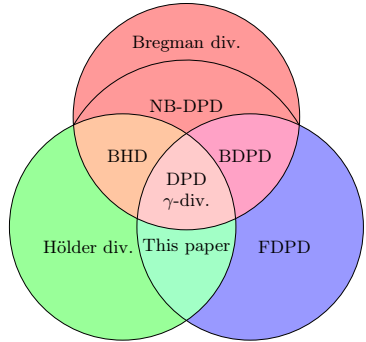
\begin{figure}[tb]
\centering
\resizebox{0.4\linewidth}{!}{
    \begin{tikzpicture}
        \def\r{2}
        \def\d{1.3}
        \begin{scope}[blend group = soft light]
            \fill[red!40!white]   ( 90:\d) circle (\r);
            \fill[green!40!white] (210:\d) circle (\r);
            \fill[blue!40!white]  (330:\d) circle (\r);
        \end{scope}
        \draw[black] (90:\d) circle(\r);
        \draw[black] (210:\d) circle(\r);
        \draw[black] (330:\d) circle(\r);
        \node at ( 90:2.6)    {Bregman div.};
        \node at ( 90:1.5)    {NB-DPD};
        \node at ( 210:2)   {H\"{o}lder div.};
        \node at ( 330:2)   {FDPD};
        \node at (0,0.2) {DPD};
        \node at (0,-0.2) {$\gamma$-div.};
        \node at (30:1.2) {BDPD};
        \node at (150:1.2) {BHD};
        \node at (270:1.0) {This paper};
        \draw[black] (30:2.3) arc (30:150:2.3);
    \end{tikzpicture} }
    \caption{Relationship among the DPD, the $\gamma$-divergence, and their generalized divergence families.}
    \label{fig:relation_div}
\end{figure}
In robust inference, the density power divergence (DPD), also known as $\beta$-divergence, is widely used \cite{beta-div}.
The DPD has a power parameter $\gamma\geq0$ that controls the trade-off between model efficiency and robustness against outliers.
The $\gamma$-divergence \cite{gamma-div,type0-div} is another density-power-based divergence.
Similarly to the DPD, the $\gamma$-divergence provides a trade-off between efficiency and robustness.
The $\gamma$-divergence can reduce the latent bias to zero, even when the proportion of outliers in the data is large \cite{gamma-div}.

These divergences are defined to ensure that the true distribution can be replaced by the empirical distribution, expressed as a sum of Dirac delta functions.
Such divergences are called non-kernel \cite{Jana2019} or decomposable \cite{Broniatowski2012}, and this property is important for practical statistical inference.
Motivated by the same consideration, the divergence induced by a composite scoring rule \cite{holder-div1} has been proposed and will be explained in Section \ref{sec:notation}.
Furthermore, several generalized divergences that include the DPD and the $\gamma$-divergence as special cases have been developed in the class of non-kernel divergences.
Examples of two- or three-parameter divergence families include the JHHB divergence family \cite{type0-div}, the Bregman--H\"{o}lder divergence (BHD) \cite{holder-div1}, and the bridge density power divergence (BDPD) \cite{bridge-div,Gayen2024}.

The DPD and the $\gamma$-divergence can be generalized into three classes of divergences with functional degrees of freedom, within which the two- and three-parameter divergence families are positioned as special subclasses (Fig. \ref{fig:relation_div}).
The first class is the (functional or non-separable) Bregman divergence \cite{Frigyik2008}, which is defined using a convex functional.
This divergence can be represented using proper scoring rules, and the estimation problem can be reduced to M-estimation \cite{Gneiting2007}.
Recently, we have proposed a norm-based Bregman density power divergence (NB-DPD), which is subclass of Bregman divergences characterized by density-power-based formulations \cite{mypaper7}.
The second class is the H\"{o}lder divergence \cite{holder-div1}, which is characterized by invariance under affine transformations.
The third class is the functional density power divergence (FDPD) \cite{FDPD}, which is defined based on the mathematical structure shared by the DPD and the $\gamma$-divergence.
Specifically, the FDPD is defined by applying a functional transformation to each term of the DPD.
Ray et al. \cite{FDPD} clarified the necessary and sufficient conditions for the function such that the FDPD is valid as a divergence.
Investigating the intersection of different divergence families causes a deeper understanding and characterization of their structural properties.
Kanamori and Fujisawa \cite{holder-div1} derived the BHD as the intersection of the Bregman divergence and the H\"{o}lder divergence.
Recently, we demonstrated that the intersection of the Bregman divergence and the FDPD is the BDPD \cite{mypaper7}.
However, the relationship between the H\"{o}lder divergence and the FDPD has not been discussed.

In this study, we discuss the relationship between the H\"{o}lder divergence and the FDPD from two perspectives.
First, we prove that the intersection of the H\"{o}lder divergence and the FDPD is limited to the JHHB divergence family (Fig. \ref{fig:relation_div}).
Furthermore, we derive the function $\eta$ corresponding to the H\"{o}lder divergence when the JHHB divergence family is expressed in the form of a H\"{o}lder divergence.
Subsequently, we define a $\xi$-H\"{o}lder divergence that includes both the H\"{o}lder divergence and the FDPD as special cases.
Furthermore, we derive an inequality between composite scoring rules corresponding to different FDPDs based on the $\xi$-H\"{o}lder divergence.
Finally, we prove that the $\xi$-H\"{o}lder divergence can be derived by imposing the mathematical structure of the H\"{o}lder score on a composite scoring rule.
The proofs of the theorems omitted in the main text are presented in the appendix.

\section{Notation and definitions}\label{sec:notation}
In this section, we introduce certain mathematical definitions.
The set of nonnegative real numbers is denoted by $\mathbb{R}_+$.
Let $X \subseteq \mathbb{R}^d$ be a subset of the Euclidean space, and define a measure space $(X, \mathcal{B}, \nu)$.
The set of nonnegative functions is denoted by $\mathcal{L}_0=\{f:X\to\mathbb{R}\mid f \text{ is measurable on } (X,\mathcal{B},\nu), f\geq0, f\neq0\}$. 
For $f\in\mathcal{L}_0$, we write the integral with respect to the Lebesgue measure $\nu$ as $\langle f\rangle = \int_X f(x)d\nu(x)$.
Suppose that $\mathcal{F}\subseteq\mathcal{L}_0$, and define the set of probability distributions as $\mathcal{P}=\{p\in\mathcal{F}\mid \langle p\rangle=1\}$.
We assume $\mathcal{F}_\gamma=\{f\in\mathcal{L}_0|\langle f^{1+\gamma}\rangle<\infty\}$ for a fixed $\gamma\geq0$.
Here, we define a specific form of composite scoring rule discussed in \cite{holder-div1}.
\begin{definition}[Composite scoring rule and divergence \cite{holder-div1}]\label{def:divergence}
    Let $U,V:\mathbb{R}_+\to\mathbb{R}$, and $T:\mathbb{R}^2\to\mathbb{R}$ be given functions.
    The composite scoring rule $S:\mathcal{F}\times\mathcal{F}\to\mathbb{R}$ is defined as 
    \begin{align*}
        S(g,f) = T\left(\langle gU(f)\rangle, \langle V(f)\rangle\right).
    \end{align*}
    The composite scoring rule $S$ is strictly proper if and only if the following conditions hold: 
    \begin{align*}
        &\forall g,f\in\mathcal{F},\; S(g,f)\geq S(g,g), \\
        &\forall q,p\in\mathcal{P},\;S(q,p) = S(q,q),\; \text{if and only if}\; q=p\; \text{almost everywhere}.
    \end{align*}
    A divergence is defined as the difference between composite scoring rules.
    From the properties of composite scoring rule, the following properties of the divergence hold:
    \begin{align*}
        &\forall g,f\in\mathcal{F},\; D(g,f)=S(g,f)-S(g,g)\geq0,\\
        &\forall q,p\in\mathcal{P},\; D(q,p)=0\; \text{if and only if}\; q=p\; \text{almost everywhere}.
    \end{align*}
\end{definition}

\begin{definition}[Equivalence of composite scoring rules \cite{holder-div1}]
    Two composite scoring rules $\tilde{S}(g,f)$ and $S(g,f)$ are equivalent if and only if there exists a strictly increasing function $\tau:\mathbb{R}\to\mathbb{R}$ such that $\tilde{S}(g,f)=\tau(S(g,f)), \;\mathrm{for\; all}\; g,f\in\mathcal{F}$.
    This equivalence also holds for probability distributions $q,p\in\mathcal{P}$.
\end{definition}

The estimator is preferably designed to transform consistently when the data is transformed.
An affine transformation of a data is expressed by $\bm{x}\mapsto\bm{\Sigma}^{-1}(\bm{x}-\bm{\mu})$, where $\bm{\Sigma}\in\mathbb{R}^{d\times d}$ is an invertible matrix and $\bm{\mu}\in\mathbb{R}^d$ is a $d$-dimensional vector. 
The corresponding probability distribution transforms as $q(\bm{x}) \mapsto q_{\bm{\Sigma},\bm{\mu}}(\bm{x})=|\det \bm{\Sigma}|q(\bm{\Sigma}\bm{x}+\bm{\mu})$.
When the statistical model $\hat{p}$ estimated from the original data distribution $q$ is equal to the model $\widehat{p_{\bm{\Sigma},\bm{\mu}}}$ estimated from the affinely transformed distribution, i.e., when $\hat{p}=\widehat{p_{\bm{\Sigma},\bm{\mu}}}$ holds, the estimator is called an affine invariant estimator \cite{holder-div1}.
A divergence that yields an affine invariant estimator is referred to as an affine invariant divergence.
Here, we define the affine invariant divergence on the space of nonnegative functions.
\begin{definition}[Affine invariant divergence \cite{holder-div1}]\label{def:affine_invariance}
    A divergence $D$ is affine invariant if there exists a real-valued function $h$ such that
    \begin{align*}
        h(\bm{\Sigma}, \bm{\mu})D(g_{\bm{\Sigma},\bm{\mu}}, f_{\bm{\Sigma},\bm{\mu}})=D(g,f).
    \end{align*}
\end{definition}

\section{Related Divergences}\label{sec:related_divergences}
\subsection{H\"{o}lder divergence}
\begin{definition}[H\"{o}lder score and divergence \cite{holder-div1}]
Let $\gamma\geq0$.
When $\gamma>0$, suppose that $\eta:\mathbb{R}_+\to\mathbb{R}$ satisfies $\eta(1)=-1$ and $\eta(z)\geq -z^{1+\gamma}$ for all $z\geq0$.
Subsequently, the H\"{o}lder score between nonnegative functions $g,f\in\mathcal{F}_\gamma$\footnote[1]{\label{fn:nonnegative_set} If $\gamma=0$, $\mathcal{F}_0$ is considered with an additional integrability condition: $\langle g\log f\rangle<\infty$.} is defined by
\begin{align}
    S_{\eta,\gamma}(g,f)=\begin{cases}
        \eta\left(\frac{\left\langle gf^\gamma\right\rangle}{\left\langle f^{1+\gamma}\right\rangle}\right)\left\langle f^{1+\gamma}\right\rangle, &(\gamma>0), \label{eq:H_score}\\
        -\langle g\log f\rangle+\langle f\rangle, &(\gamma=0).
    \end{cases}
\end{align}
The H\"{o}lder divergence is defined as the difference of the H\"{o}lder scores:
\begin{align*}
    D_{\eta, \gamma}(g,f) = \begin{cases}
        \eta\left(\frac{\left\langle gf^\gamma\right\rangle}{\left\langle f^{1+\gamma}\right\rangle}\right)\left\langle f^{1+\gamma}\right\rangle + \left\langle g^{1+\gamma}\right\rangle, &(\gamma>0),\\
        \left\langle g\log\frac{g}{f}\right\rangle-\langle g\rangle+\langle f\rangle, &(\gamma=0).
    \end{cases}
\end{align*}
\end{definition}
The H\"{o}lder divergence reduces to several known divergences depending on the choice of $\eta$.
Specifically, if $\eta(z)=\gamma-(1+\gamma)z$, it becomes the DPD \cite{beta-div};
if $\eta(z)=-z^{1+\gamma}$, it generates the pseudo-spherical (PS) type $\gamma$-divergence \cite{gamma-div,type0-div}.
For $\eta(z)=-|\kappa z-\kappa+1|^{\frac{1+\gamma}{\kappa}}\mathrm{sign}(\kappa z-\kappa+1)$, $\kappa\geq1$, the H\"{o}lder divergence corresponds to the BHD \cite{holder-div1}, where $\mathrm{sign}(z)=z/|z|$ denotes the sign function with $\mathrm{sign}(0)=0$.
This divergence reduces to the DPD when $\kappa=1+\gamma$, and to the PS type $\gamma$-divergence when $\kappa=1$.
Although this relationship has not been demonstrated in the existing literature, we will show in Section \ref{sec:intersection} that the JHHB divergence family \cite{type0-div} is a subclass of the H\"{o}lder divergence.

\subsection{Functional density power divergence}
Ray et al. \cite{FDPD} introduced the FDPD between probability distributions.
In this section, we extend the definition of the FDPD to nonnegative functions.
\begin{definition}[Functional density power score and divergence \cite{FDPD}]\label{def:FDPD}
Let $\psi:[-\infty,\infty)\to[-\infty,\infty]$ be a continuous, strictly increasing, and convex function, as well as define $\varphi:[0,\infty)\to[-\infty,\infty]$ by $\psi(z)=\varphi(e^z)$.
For $\gamma\geq0$, define FDPD between nonnegative functions $g,f\in\mathcal{F}_\gamma$\textsuperscript{\ref{fn:nonnegative_set}} as
\begin{subnumcases}{D_{\varphi,\gamma}(g,f) =}
    \frac{1}{\gamma}\varphi\left(\left\langle g^{1+\gamma}\right\rangle\right) - \frac{1+\gamma}{\gamma}\varphi\left(\left\langle gf^\gamma\right\rangle\right)+\varphi\left(\left\langle f^{1+\gamma}\right\rangle\right), &$(\gamma>0)$, \label{eq:FDPD}\\
    \varphi'(\langle g\rangle)\left\langle g\log\frac{g}{f}\right\rangle -\varphi(\langle g\rangle)+\varphi(\langle f\rangle), &$(\gamma=0)$, \label{eq:FDPD_KL}
\end{subnumcases}
where $\varphi'$ denotes the derivative of $\varphi$.
The corresponding functional density power score (FDPS) is denoted by
\begin{subnumcases}{S_{\varphi,\gamma}(g,f)=}
    \gamma\varphi\left(\left\langle f^{1+\gamma}\right\rangle\right)- (1+\gamma)\varphi\left(\left\langle gf^\gamma\right\rangle\right), &$(\gamma>0)$, \label{eq:FDP_score_gamma}\\
    -\varphi'(\langle g\rangle)\left\langle g\log f\right\rangle +\varphi(\langle f\rangle), &$(\gamma=0)$. \label{eq:FDP_score_0}
\end{subnumcases}
\end{definition}
\begin{remark}\label{remark:FDP_score}
    When $\gamma>0$, \eqref{eq:FDP_score_gamma} defines a composite scoring rule for any function $\varphi$ with $U(z)=z^\gamma$, $V(z)=z^{1+\gamma}$, and $T(x,y)=-(1+\gamma)\varphi(x)+\gamma\varphi(y)$ in Definition \ref{def:divergence}.
    However, when $\gamma = 0$, \eqref{eq:FDP_score_0} does not yield a composite scoring rule unless $\varphi'(z)$ is constant.
\end{remark}
The FDPD \eqref{eq:FDPD_KL} at $\gamma=0$ can be derived as the limit of $\gamma\to0$.
The FDPD remains invariant under affine transformations of $\varphi$, that is, $D_{a\varphi+b,\gamma}(g,f)=aD_{\varphi,\gamma}(g,f)$, where $a>0$ and $b\in\mathbb{R}$ are constants.
When $\varphi(z)=z$, the FDPD becomes the DPD; when $\varphi(z)=\log z$, it becomes the $\gamma$-divergence.
For $\varphi(z)=(z^\zeta-1)/\zeta$, $(\zeta>0)$, it reduces to JHHB divergence family \cite{type0-div}, and for $\varphi(z)=\log(\lambda_1+\lambda_2 z)/\lambda_2, (\lambda_1\geq0,\lambda_2>0)$, it becomes the BDPD \cite{bridge-div,Gayen2024}.

\section{Intersection of H\"{o}lder divergence and FDPD}\label{sec:intersection}
In this section, we prove that the intersection of the H\"{o}lder divergence and the FDPD corresponds to a generalized divergence family introduced by Jones et al. \cite{type0-div}.
Based on Definition \ref{def:affine_invariance}, we derive the function $\varphi$ for which the FDPD satisfies affine invariance, and obtain the following theorem.
\begin{theorem}\label{thm:fdpd_affine}
The FDPD is affine invariant if and only if the function $\varphi$ is given by $(z^\zeta-1)/\zeta$ for $\zeta>0$, or $\log z$.
The scale function of the affine transformation is given by $h(\bm{\Sigma},\bm{\mu}) =|\det \bm{\Sigma}|^{-\gamma\zeta}$ for $\gamma>0$ or $h(\bm{\Sigma},\bm{\mu}) =|\det \bm{\Sigma}|^{-\zeta}$ for $\gamma=0$.
\end{theorem}
From Theorem \ref{thm:fdpd_affine}, the affine invariant FDPD is denoted by
\small
\begin{subnumcases}{D_{\zeta,\gamma}(g,f) =}
    \frac{1}{\gamma\zeta}\left\langle g^{1+\gamma}\right\rangle^\zeta - \frac{1+\gamma}{\gamma\zeta}\left\langle gf^\gamma\right\rangle^\zeta + \frac{1}{\zeta}\left\langle f^{1+\gamma}\right\rangle^\zeta, &$(\gamma>0, \zeta>0)$, \label{eq:JHHB_zeta}\\
    \frac{1}{\gamma}\log\left\langle g^{1+\gamma}\right\rangle - \frac{1+\gamma}{\gamma}\log\left\langle gf^\gamma\right\rangle + \log\left\langle f^{1+\gamma}\right\rangle,&$(\gamma>0,\zeta=0)$, \label{eq:zeta_gamma_div}\\
    \langle g\rangle^{\zeta-1}\left\langle g\log\frac{g}{f}\right\rangle - \frac{1}{\zeta}\langle g\rangle^\zeta + \frac{1}{\zeta}\langle f\rangle^\zeta, &$(\gamma=0,\zeta>0),$ \label{eq:JHHB_KL_zeta}\\
    \frac{1}{\langle g\rangle}\left\langle g\log\frac{g}{f}\right\rangle -\log\langle g\rangle + \log \langle f\rangle , &$(\gamma=0,\zeta=0).$ \label{eq:JHHB_KL}
\end{subnumcases}
\normalsize
In this study, we refer to this divergence as the JHHB divergence family.
Eqs. \eqref{eq:JHHB_zeta} and \eqref{eq:zeta_gamma_div} were introduced by Jones et al. \cite{type0-div} as a divergence that bridges the DPD ($\zeta=1$) \cite{beta-div} and the $\gamma$-divergence ($\zeta=0$) \cite{gamma-div,type0-div}.
From Theorem \ref{thm:fdpd_affine} and Theorem 4.2 in \cite{holder-div1}, the following corollary holds.
\begin{corollary}
    Under Assumption of Theorem 4.2 in \cite{holder-div1}, the intersection of the H\"{o}lder divergence and the FDPD is limited to the JHHB divergence family, specifically \eqref{eq:JHHB_zeta}, \eqref{eq:zeta_gamma_div}, and \eqref{eq:JHHB_KL_zeta} (with $\zeta=1$ for \eqref{eq:JHHB_KL_zeta} only).
\end{corollary}
\begin{proof}
    According to Theorem 4.2 in \cite{holder-div1}, any affine invariant proper composite scoring rule is equivalent to the H\"{o}lder score.
    From Theorem \ref{thm:fdpd_affine} and Remark \ref{remark:FDP_score}, the affine invariant FDPD that is represented by the proper composite scoring rule is expressed in \eqref{eq:JHHB_zeta}, \eqref{eq:zeta_gamma_div}, and \eqref{eq:JHHB_KL_zeta}, with $\zeta=1$ applying to \eqref{eq:JHHB_KL_zeta} only.
    Therefore, the intersection of the H\"{o}lder divergence and the FDPD is limited to JHHB divergence family, specifically \eqref{eq:JHHB_zeta}, \eqref{eq:zeta_gamma_div}, and \eqref{eq:JHHB_KL_zeta} with $\zeta=1$.
    \qed
\end{proof}
The following theorem provides the function $\eta$ corresponding to the case where the JHHB divergence family is expressed as a H\"{o}lder divergence.
\begin{theorem}\label{thm:hd_jones_div}
For $\gamma>0$ and $\zeta\geq0$, the JHHB divergence families \eqref{eq:JHHB_zeta} and \eqref{eq:zeta_gamma_div} can be represented as a H\"{o}lder divergence with the generating function
\begin{align*}
    \eta(z)=
    \begin{cases}
        -\left|(1+\gamma)z^\zeta-\gamma\right|^{\frac{1}{\zeta}}\cdot\mathrm{sign}\left((1+\gamma)z^\zeta-\gamma\right), &(\zeta>0),\\
        -z^{1+\gamma}, &(\zeta=0).
    \end{cases}
\end{align*}
\end{theorem}

\section{Generalization of H\"{o}lder divergence and FDPD}\label{sec:generalization}
The nonnegativity of the H\"{o}lder divergence and the FDPD is established via a two-step inequality: the first step involves a condition on the function $\eta$ or $\varphi$, and the second relies on H\"{o}lder’s inequality.
Considering that both divergences have their nonnegativity established through H\"{o}lder’s inequality, we define their generalization as follows.
\begin{definition}[$\xi$-H\"{o}lder  score and divergence]
    Let $\gamma>0$, and let $\eta:\mathbb{R}_+\to\mathbb{R}$ be a function satisfying $\eta(1)=-1$ and $\eta(z)\geq -z^{1+\gamma}$ for all $z\geq0$.
    Let $\psi:[-\infty,\infty)\to[-\infty,\infty]$ be a continuous, strictly increasing, and convex function, and define $\xi:[0,\infty)\to[0,\infty]$ by $\xi(z)=\exp(\psi(\log(z)))$.
    The $\xi$-H\"{o}lder score between nonnegative functions $g,f\in\mathcal{F}_\gamma$ is defined as
    \begin{align}
        S_{\eta,\xi,\gamma}(g,f)=\eta\left(\frac{\xi(\langle gf^\gamma\rangle)}{\xi(\langle f^{1+\gamma}\rangle)}\right)\xi(\langle f^{1+\gamma}\rangle). \label{eq:3HS}
    \end{align}
    The $\xi$-H\"{o}lder divergence is defined as the difference of the $\xi$-H\"{o}lder scores:
    \begin{align}
        D_{\eta,\xi,\gamma}(g,f) = \eta\left(\frac{\xi(\langle gf^\gamma\rangle)}{\xi(\langle f^{1+\gamma}\rangle)}\right)\xi(\langle f^{1+\gamma}\rangle) + \xi(\langle g^{1+\gamma}\rangle). \label{eq:3HD}
    \end{align}
\end{definition}
The nonnegativity of the $\xi$-H\"{o}lder divergence \eqref{eq:3HD} is established as follows:
\begin{align*}
    \eta\left(\frac{\xi(\langle gf^\gamma\rangle)}{\xi(\langle f^{1+\gamma}\rangle)}\right)\xi(\langle f^{1+\gamma}\rangle)\overset{(a)}{\geq} -\frac{\xi(\langle gf^\gamma\rangle)^{1+\gamma}}{\xi(\langle f^{1+\gamma}\rangle)^\gamma}\overset{(b)}{\geq} -\xi(\langle g^{1+\gamma}\rangle).
\end{align*}
Inequality (a) follows from the condition $\eta(z)\geq-z^{1+\gamma}$, while inequality (b) is derived from the two-step inequality used in the FDPD, which is based on the strict monotonicity and convexity of $\psi(z)=\log \xi(e^z)$, along with H\"{o}lder’s inequality.
When $\xi(z)=z$, the $\xi$-H\"{o}lder score \eqref{eq:3HS} reduces to the H\"{o}lder score.
Furthermore, when $\eta(z)=\gamma-(1+\gamma)z$ and $\xi(e^z)$ is strictly increasing and convex, the $\xi$-H\"{o}lder score \eqref{eq:3HS} reduces to the FDPS \eqref{eq:FDP_score_gamma} defined by $\varphi(z)=\xi(z)$.
Similarly, when $\eta(z)=-z^{1+\gamma}$ and $\log \xi(e^z)$ is strictly increasing and convex, the $\xi$-H\"{o}lder score \eqref{eq:3HS} reduces to the FDPS \eqref{eq:FDP_score_gamma} defined by $\varphi(z)=\log \xi(z)$.
Because $\gamma-(1+\gamma)z\geq-z^{1+\gamma}$, the FDPS has lower bound as described in the following theorem.
\begin{theorem}\label{thm:fdp_inq}
    The FDPS \eqref{eq:FDP_score_gamma} defined by $\varphi$ has the lower bound:
    \begin{align*}
        \gamma \varphi(\langle f^{1+\gamma}\rangle) - (1+\gamma)\varphi(\langle gf^\gamma\rangle)\geq -\exp\left(-\left[\gamma \varphi_*(\langle f^{1+\gamma}\rangle) - (1+\gamma)\varphi_*(\langle gf^\gamma\rangle)\right]\right),
    \end{align*}
    where $\varphi_*(z)=\log \varphi(z)$.
    The lower bound is an FDPS defined by $\varphi_*$ if and only if $\varphi(e^z)$ is a strictly increasing and log-convex function.
\end{theorem}
This inequality can be regarded as a generalization with respect to $\varphi$ of the inequality between the density power score and the PS score, which is recovered when $\varphi(z)=z$.

We prove that the $\xi$-H\"{o}lder score \eqref{eq:3HS} can be derived by imposing the mathematical structure of the H\"{o}lder score on a proper composite scoring rule.
The H\"{o}lder score \eqref{eq:H_score} has the property that when the argument of the function $\eta$ is equal to one, i.e., when $\langle g f^\gamma \rangle = \langle f^{1+\gamma} \rangle$, the coefficient $\langle f^{1+\gamma} \rangle$ in $\eta$ represents an (negative) H\"{o}lder entropy, because $S(g,g)=-\langle g^{1+\gamma}\rangle$.
Generalizing this idea, we express the generalized H\"{o}lder score using functions $\bar{u},\bar{v}$ as follows 
\begin{align*}
    \bar{S}(g,f) = \eta\left(\frac{\bar{u}(g,f)}{\bar{v}(f)}\right)\bar{v}(f).    
\end{align*}
If we require that $\bar{S}(g,g) = -\bar{v}(g)$, it must follow that 
\begin{align*}
    \forall g\in\mathcal{F},\; \bar{u}(g,g) = \bar{v}(g).
\end{align*}
Assuming that the generalized H\"{o}lder score is a composite scoring rule (Definition \ref{def:divergence}), the functions $\bar{u}$ and $\bar{v}$ must be representable in the form $\bar{u}(g,f)=u(\langle gU(f)\rangle)$ and $\bar{v}(f)=v(\langle V(f)\rangle)$ for functions $u$ and $v$, respectively.
Accordingly, we impose the following assumption.
\begin{assumption}\label{assumption:holder_form}
Let $u$ and $v$ be strictly increasing and continuous functions.
For $\gamma>0$, let $\eta:\mathbb{R}_+\to\mathbb{R}$ be a function such that $\eta(1) = -1$ and $\eta(z) \geq -z^{1+\gamma}$ holds for all $z \geq 0$.
We assume that the composite scoring rule can be expressed as follows:
\begin{align}
    S(g, f) = \eta\left(\frac{u\left(\langle gU(f)\rangle\right)}{v\left(\langle V(f)\rangle\right)}\right) v\left(\langle V(f)\rangle\right), \label{eq:H_assump}
\end{align}
where $u,v$ are functions on $\mathbb{R}$.
Furthermore, we suppose that for any nonnegative function $g\in\mathcal{F}$ such that both integrals are finite, the following holds:
\begin{align*}
    u(\langle gU(g)\rangle) = v(\langle V(g)\rangle).
\end{align*}
\end{assumption}
\begin{theorem}\label{thm:formula_V_and_U}
    Under Assumptions 4.1 and 4.2 in \cite{holder-div1}, and Assumption \ref{assumption:holder_form}, the functions $U$ and $V$ are denoted by $U(z)=cz^\gamma$ and $V(z)=c/az^{1+\gamma}$, respectively, for $\gamma>0$, where $a,c\in\mathbb{R}\setminus\{0\}$ are constants.
    The functions $u$ and $v$ satisfy the relation $u(az)=v(z)$.
    Subsequently, the composite scoring rule is given by
    \begin{align}
        S(g,f) = \eta\left(\frac{u(c\langle gf^\gamma\rangle)}{u(c\langle f^{1+\gamma}\rangle)}\right)u(c\langle f^{1+\gamma}\rangle). \label{eq:H_FDP_score}
    \end{align}
\end{theorem}
By Theorem \ref{thm:formula_V_and_U}, it has been shown that the composite scoring rule is expressed in the form \eqref{eq:H_FDP_score}.
However, it remains unclear what conditions on a constant $c\in\mathbb{R}\setminus\{0\}$ and the function $u$ are necessary for \eqref{eq:H_FDP_score} to define a strictly proper composite scoring rule.
The following theorem clarifies this point.
\begin{theorem}\label{thm:H_FDP_nonnegative}
    Let $\psi(z)=\log u(e^z)$.
    The composite scoring rule \eqref{eq:H_FDP_score} is strictly proper if and only if $\psi:(-\infty,\infty)\to[-\infty,\infty]$ is a strictly increasing and convex function, $u:[0,\infty)\to[0,\infty]$, and $c$ is a positive constant.
\end{theorem}
Thus, by setting $c=1$ in \eqref{eq:H_FDP_score} and replacing $u$ with $\xi$, we obtain the $\xi$-H\"{o}lder score \eqref{eq:3HS}.

\section{Conclusion and future directions}\label{sec:conclusion}
In this study, we discussed the relationship between the H\"{o}lder divergence and the FDPD from two perspectives.
First, we proved that the intersection of H\"{o}lder divergence and the FDPD is limited to the JHHB divergence family.
Second, we focused on the fact that H\"{o}lder's inequality guarantees the nonnegativity of both the H\"{o}lder divergence and the FDPD, and constructed the $\xi$-H\"{o}lder divergence based on this property.
Furthermore, we proved that the $\xi$-H\"{o}lder divergence can be derived by imposing the mathematical structure of the H\"{o}lder score on composite scoring rule.
Future directions include extending the $\xi$-H\"{o}lder divergence to negative $\gamma$ \cite{holder-div2} and to non-composite score based divergences, such as the $\alpha\beta$-divergence families \cite{Cichocki2010,Nielsen2017}, which are characterized by H\"{o}lder's inequality.

\begin{credits}
\subsubsection{\ackname}
This work was supported by JSPS KAKENHI Grant numbers JP23K16849.
\subsubsection{\discintname}
The author has no competing interests to declare that are relevant to the content of this article.
\end{credits}
\bibliographystyle{splncs04}
\bibliography{Refj}

\clearpage
\appendix
\section{Conditions for FDPD on nonnegative functions}\label{sec:FDPD_non_negative}
Here, we extend Propositions 4.1 and 4.2 in \cite{FDPD}, originally established for probability distributions, to the setting of nonnegative functions.
With minor modifications to the original proofs, analogous results can be generally derived as follows.
\begin{theorem}\label{thm:fdpd_nonnegative}
    Let $\gamma>0$.
    For all nonnegative functions $g,f\in\mathcal{F}_\gamma$, $D_{\varphi,\gamma}(g,f)\geq0$ holds.
    If $\psi$ is a strictly convex function, $D_{\varphi,\gamma}(g,f)=0$ holds if and only if $g=f$ almost everywhere.
    However, if $\psi$ is not a strictly convex, $D_{\varphi,\gamma}(g,f)=0$ holds when $g=f$ almost everywhere or when both of the following conditions are satisfied:
    \begin{align}
        t\psi(\log\langle g^{1+\gamma}\rangle) + (1 - t)\psi(\log\langle f^{1+\gamma}\rangle) = \psi\left(t \log \langle g^{1+\gamma}\rangle + (1 - t)\log\langle f^{1+\gamma}\rangle\right), \label{eq:FDPD_equality_cond}
    \end{align}
    and $g^{1+\gamma}=cf^{1+\gamma}$ almost everywhere, where $t=1/(1+\gamma)$ and $c>0$.
    For all probability distributions $q,p\in\mathcal{P}_\gamma$, $D_{\varphi,\gamma}(q,p)=0$ holds if and only if $q=p$ almost everywhere.
\end{theorem}
\begin{proof}
    From \eqref{eq:FDPD}, we have
    \begin{align*}
        &D_{\varphi,\gamma}(g,f)\\
        &=\frac{1}{\gamma}\varphi\left(\left\langle g^{1+\gamma}\right\rangle\right) - \frac{1+\gamma}{\gamma}\varphi\left(\left\langle gf^\gamma\right\rangle\right)+\varphi\left(\left\langle f^{1+\gamma}\right\rangle\right)\\
        &= \frac{1+\gamma}{\gamma}\left[\frac{1}{1+\gamma}\psi\left(\log\left\langle g^{1+\gamma}\right\rangle\right)+\frac{\gamma}{1+\gamma}\psi\left(\log\left\langle f^{1+\gamma}\right\rangle\right)\right]-\frac{1+\gamma}{\gamma}\psi\left(\log\left\langle gf^\gamma \right\rangle\right)\\
        &\overset{(a)}{\geq} \frac{1+\gamma}{\gamma}\psi\left(\frac{1}{1+\gamma}\log\left\langle g^{1+\gamma}\right\rangle + \frac{\gamma}{1+\gamma}\log\left\langle f^{1+\gamma}\right\rangle\right) - \frac{1+\gamma}{\gamma}\psi\left(\log\left\langle gf^\gamma \right\rangle\right)\\
        &\overset{(b)}{\geq} 0.
    \end{align*}
    Inequality (a) follows from the convexity of $\psi$, with equality holding when $\psi$ is an affine function over an interval of its domain or $\langle g^{1+\gamma}\rangle=\langle f^{1+\gamma}\rangle$.
    Inequality (b) is derived from the strict monotonicity of $\psi$ and H\"{o}lder's inequality, with equality holding when $g^{1+\gamma}=cf^{1+\gamma}$ almost everywhere for positive constant $c>0$.
    Therefore, for $D_{\varphi,\gamma}(g,f)=0$ to hold, equality must be satisfied in both inequalities (a) and (b).
    When $g$ and $f$ are probability distributions or $\psi$ is a strictly convex function, equality in both inequalities (a) and (b) holds if and only if $g=f$ almost everywhere.
    However, if $\psi$ is not strictly convex, it contains piecewise linear segments within its domain.
    Because $\psi$ is strictly increasing function, piecewise constant functions are excluded.
    The equality $D_{\varphi,\gamma}(g,f)=0$ holds if and only if either $g=f$ almost everywhere, or $g^{1+\gamma}=cf^{1+\gamma}$ almost everywhere and \eqref{eq:FDPD_equality_cond} is satisfied.
    \qed
\end{proof}
\begin{theorem}\label{thm:fdpd_converse}
    Let $\gamma>0$ and let $D_{\varphi,\gamma}(g,f)$ be the FDPD defined for nonnegative functions by $\varphi$.
    For all nonnegative functions $g,f\in\mathcal{F}_\gamma$, we assume that $D_{\varphi,\gamma}(g,f)\geq0$ and that $D_{\varphi,\gamma}(g,f)>0$ whenever the supports of $g$ and $f$ are different. 
    Then, the function $\psi(z)=\varphi(e^z)$ is strictly increasing and convex.
\end{theorem}
\begin{proof}
    The proof in \cite{FDPD} derived the conditions for $\psi$ by substituting specific probability density functions, under the assumption that for any probability densities $q,p\in\mathcal{P}_\gamma$, $D_{\varphi,\gamma}(q,p)\geq0$ and $D_{\varphi,\gamma}(q,p)=0\Leftrightarrow q=p$ almost everywhere.
    Specifically, two normalized indicator functions were substituted with disjoint supports, based on $D_{\varphi,\gamma}(q,p)\neq0\Leftrightarrow q\neq p$.
    Under the corresponding assumption that the divergence is positive when the supports of two nonnegative functions are disjoint, their proof can be applied in the same way by replacing the probability densities with scalar multiples $g=aq$ and $f=ap$, where $a>0$ is a constant.
    \qed
\end{proof}
\begin{corollary}\label{coro:FDPD}
Let $\gamma>0$.
The FDPD \eqref{eq:FDPD} satisfies the definition of a divergence for nonnegative functions if and only if $\psi$ is a strictly increasing and convex function. 
Specifically, the following conditions hold:
    \begin{align*}
        &\forall g,f\in\mathcal{F}_\gamma,\quad D_{\varphi,\gamma}(g,f)\geq0,\\
        &\forall q,p\in\mathcal{P}_\gamma,\quad D_{\varphi,\gamma}(q,p)=0\; \text{holds almost everywhere if and only if}\; q=p.
    \end{align*}
\end{corollary}
\begin{proof}
    According to Theorem \ref{thm:fdpd_nonnegative}, the condition $D_{\varphi,\gamma}(g,f)=0$ holds only if at least one of the following is satisfied:
    \begin{itemize}
        \item $g=f$ almost everywhere,
        \item $g^{1+\gamma}=cf^{1+\gamma}$ almost everywhere for some constant $c>0$.
    \end{itemize}
    In either case, the supports of the nonnegative functions $g$ and $f$ must coincide.
    Thus, the assumptions of Theorem \ref{thm:fdpd_converse} are satisfied, and Corollary \ref{coro:FDPD} follows directly from Theorems \ref{thm:fdpd_nonnegative} and \ref{thm:fdpd_converse}.
    \qed
\end{proof}

\section{Proof of Theorem \ref{thm:fdpd_affine}}
\subsection{$\gamma>0$}
We assume that FDPD \eqref{eq:FDPD} is an affine invariant divergence.
From Definition \ref{def:affine_invariance}, it follows that the following equation must hold.
\begin{align}
    \begin{split}        
        &\frac{1}{\gamma}\varphi\left(|\det \bm{\Sigma}|^\gamma\left\langle g^{1+\gamma}\right\rangle\right) - \frac{1+\gamma}{\gamma}\varphi\left(|\det \bm{\Sigma}|^\gamma\left\langle gf^\gamma\right\rangle\right)+\varphi\left(|\det \bm{\Sigma}|^\gamma\left\langle f^{1+\gamma}\right\rangle\right)\\
        &=h(\bm{\Sigma}, \bm{\mu})^{-1}\left[\frac{1}{\gamma}\varphi\left(\left\langle g^{1+\gamma}\right\rangle\right) - \frac{1+\gamma}{\gamma}\varphi\left(\left\langle gf^\gamma\right\rangle\right)+\varphi\left(\left\langle f^{1+\gamma}\right\rangle\right)\right]
    \end{split} \label{eq:FDPD_FEQ}
\end{align}
From \eqref{eq:FDPD_FEQ}, for any $\bm{\Sigma}_1 \neq \bm{\Sigma}_2$ satisfying $|\det \bm{\Sigma}_1|^\gamma = |\det \bm{\Sigma}_2|^\gamma$, substituting $\bm{\Sigma}_1$ and $\bm{\Sigma}_2$ into the left-hand side of \eqref{eq:FDPD_FEQ} results in the same value for a fixed $\gamma$.
Thus, it follows that $h(\bm{\Sigma}_1, \bm{\mu}) = h(\bm{\Sigma}_2, \bm{\mu})$, implying that $h$ is essentially a function of $|\det \bm{\Sigma}|^\gamma$ and can be expressed using a function $\bar{h}$ as follows:
\begin{align*}
    \bar{h}(|\det\bm{\Sigma}|^\gamma, \bm{\mu}) = h(\bm{\Sigma},\bm{\mu}).
\end{align*}
For simplicity, we omit $\bm{\mu}$ in the following discussion.
Setting $A = |\det \bm{\Sigma}|^\gamma$, $X = \left\langle g^{1+\gamma} \right\rangle$, $Y = \left\langle gf^\gamma \right\rangle$, and $Z = \left\langle f^{1+\gamma} \right\rangle$, we substitute them into \eqref{eq:FDPD_FEQ} and obtain
\begin{align*}
    \bar{h}(A)\left[\frac{1}{\gamma}\varphi\left(AX\right) - \frac{1+\gamma}{\gamma}\varphi\left(AY\right)+\varphi\left(AZ\right)\right]=\frac{1}{\gamma}\varphi\left(X\right) - \frac{1+\gamma}{\gamma}\varphi\left(Y\right)+\varphi\left(Z\right).
\end{align*}
We assume that the functional equation for $\varphi$ holds for any $A,X,Y,Z\in\mathbb{R}_+$.
Thus, we solve $\varphi$ by substituting specific values into each variable.
Substituting $X = Z$, we obtain
\begin{align*}
    \bar{h}(A)\left[\varphi\left(AX\right) - \varphi\left(AY\right)\right] = \varphi\left(X\right) - \varphi\left(Y\right).
\end{align*}
We define the function as
\begin{align}
    \bar{\varphi}(X) = \varphi(X) - \varphi(1) \label{eq:u_varphi}
\end{align}
and by setting $Y = 1$, we obtain
\begin{align}
    \bar{\varphi}(AX) = \bar{\varphi}(A) + \bar{h}(A)^{-1}\bar{\varphi}(X). \label{eq:uax1}
\end{align}
Exchanging $A$ and $X$ in \eqref{eq:uax1}, we obtain
\begin{align}
    \bar{\varphi}(AX) = \bar{\varphi}(X) + \bar{h}(X)^{-1}\bar{\varphi}(A). \label{eq:uax}
\end{align}
From \eqref{eq:uax1} and \eqref{eq:uax}, the following relation holds:
\begin{align}
    \bar{\varphi}(X)\left[1-\bar{h}(A)^{-1}\right] = \bar{\varphi}(A)\left[1-\bar{h}(X)^{-1}\right]. \label{eq:identity}
\end{align}
The trivial solutions to \eqref{eq:identity} are either $\bar{\varphi}(X)=0$ for all $X$ or $\bar{h}(X)^{-1}=1$ for all $X$.
In the case where $\bar{\varphi}(X)=0$ for all $X$, it follows from \eqref{eq:u_varphi} that $\varphi(X)=\varphi(1)$, which does not generate a valid divergence.
However, if $\bar{h}(X)^{-1}=1$ for all $X$, we obtain the following functional equation from \eqref{eq:uax1}:
\begin{align*}
    \bar{\varphi}(AX) = \bar{\varphi}(A) + \bar{\varphi}(X).
\end{align*}
Under the continuity of $\bar{\varphi}$, the general solution to this equation is expressed by $\bar{\varphi}(X) = a\log X$ for some $a \in \mathbb{R}$ \cite[pp. 25--26]{Aczel}.
Thus, from \eqref{eq:u_varphi}, the solution to \eqref{eq:FDPD_FEQ} is
\begin{align*}
    \varphi(X) = a\log X + \varphi(1),
\end{align*}
where $a$ must be positive owing to the nonnegativity of the divergence.
Because $a>0$ and $\varphi(1)\geq0$ are arbitrary constants, we may set $a=1$ and $\varphi(1)=0$ without loss of generality.
Under this setting, we obtain
\begin{align}
    \varphi(z) = \log z. \label{eq:varphi_log}
\end{align}
Eq. \eqref{eq:varphi_log} generates the $\gamma$-divergence, which is the JHHB divergence family \eqref{eq:zeta_gamma_div} with $\gamma>0$ and $\zeta=0$.
Under $X\neq1$ and $A\neq1$, from \eqref{eq:identity}, we define the constant:
\begin{align}
    b = \frac{1-\bar{h}(X)^{-1}}{\bar{\varphi}(X)} = \frac{1-\bar{h}(A)^{-1}}{\bar{\varphi}(A)}, \label{eq:const_m}
\end{align}
which does not depend on $A$ and $X$.
From \eqref{eq:const_m}, we put $\Phi(X)=\bar{h}(X)^{-1} = 1 - b \bar{\varphi}(X)$. 
We obtain
\begin{align}
    \bar{\varphi}(X) = \frac{1}{b} \left(1 - \Phi(X)\right). \label{eq:ux}
\end{align}
Eq. \eqref{eq:uax} becomes
\begin{align}
    \bar{\varphi}(AX) = \bar{\varphi}(X) + \Phi(X) \bar{\varphi}(A). \label{eq:uax2}
\end{align}
Substituting \eqref{eq:ux} into \eqref{eq:uax2}, we obtain the following functional equation:
\begin{align*}
    \Phi(AX) = \Phi(A) \Phi(X).
\end{align*}
Under the continuity of $\Phi$, the general solution to this functional equation is denoted by $\Phi(X)=X^c$, where $c\in\mathbb{R}$ is an arbitrary constant \cite[pp. 28--30]{Aczel}.
Thus, from \eqref{eq:u_varphi} and \eqref{eq:ux}, the solution to \eqref{eq:FDPD_FEQ} is expressed by
\begin{align*}
    \varphi(X) = \frac{1}{b} \left[1 - X^c\right] + \varphi(1).
\end{align*}
Because, $\psi(z) = \varphi(e^z)$ is strictly increasing and convex, according to Definition \ref{def:FDPD}, the signs of $b$ and $c$ are determined as $b < 0$ and $c > 0$.
For $\zeta>0$, setting the arbitrary constants as $b = -\zeta$, $c = \zeta$, and $\varphi(1) = 0$, we obtain
\begin{align}
    \varphi(z) = \frac{z^\zeta - 1}{\zeta}, \quad (\zeta > 0). \label{eq:zeta_power_func}
\end{align}
In particular, \eqref{eq:zeta_power_func} generates the JHHB divergence family \eqref{eq:JHHB_zeta} with $\gamma>0$ and $\zeta>0$.
From $\bar{h}(A)=A^{-\zeta}$, the scale function of the affine transformation is denoted by $h(\bm{\Sigma}, \bm{\mu})=|\det\bm{\Sigma}|^{-\gamma\zeta}$ for $\zeta\geq0$ and $\gamma>0$.
\qed

\subsection{$\gamma=0$}
We assume that FDPD \eqref{eq:FDPD_KL} is an affine invariant divergence.
From Definition \ref{def:affine_invariance}, it follows that the following equation must hold:
\begin{align*}
    \begin{split}
        &h(\bm{\Sigma}, \bm{\mu})\left[|\det \bm{\Sigma}|\varphi'\left(|\det \bm{\Sigma}|\langle g\rangle\right)\left\langle g\log\frac{g}{f}\right\rangle -\varphi\left(|\det\bm{\Sigma}|\langle g\rangle\right) + \varphi\left(|\det\bm{\Sigma}|\langle f\rangle\right)\right]\\
        &=\varphi'\left(\langle g\rangle\right)\left\langle g\log\frac{g}{f}\right\rangle -\varphi(\langle g\rangle)+\varphi(\langle f\rangle).
    \end{split}
\end{align*}
Similar to the case $\gamma>0$, the function $h$ can be expressed as a function of $|\det \bm{\Sigma}|$.
In particular, there exists a function $\bar{h}$ such that $\bar{h}(|\det \bm{\Sigma}|)=h(\bm{\Sigma})$, where $\bm{\mu}$ is omitted for simplicity.
Setting $A=|\det \bm{\Sigma}|$, $X=\langle g\rangle$, $Y=\langle f\rangle$, and $Z=\langle g\log\frac{g}{f}\rangle$, we obtain
\begin{align}
    \bar{h}(A)\left[A\varphi'(AX)Z - \varphi(AX)+\varphi(AY)\right]=\varphi'(X)Z - \varphi(X)+\varphi(Y). \label{eq:FDPD_KL_FEQ}
\end{align}
Because both sides of \eqref{eq:FDPD_KL_FEQ} are linear functions of $Z$, both the coefficients of $Z$ and constant terms must be equal.
Therefore, we obtain
\begin{align}
    A\bar{h}(A)\varphi'(AX) = \varphi'(X), \label{eq:Z_coeff}
\end{align}
and
\begin{align}
    \bar{h}(A)\left[\varphi(AX)-\varphi(AY)\right] = \varphi(X)-\varphi(Y). \label{eq:Z_const}
\end{align}
From \eqref{eq:Z_coeff} and \eqref{eq:Z_const}, we have
\begin{align*}
    \frac{\varphi(X)-\varphi(Y)}{\varphi(AX)-\varphi(AY)} = \frac{\varphi'(X)}{A\varphi'(AX)}.
\end{align*}
Setting $X=1$, we obtain
\begin{align*}
    \varphi(1)-\varphi(Y) = \frac{\varphi'(1)}{A\varphi'(A)}(\varphi(A)-\varphi(AY)).
\end{align*}
Differentiating both sides with respect to $Y$, we obtain
\begin{align}
    \varphi'(AY)\varphi'(1) = \varphi'(A)\varphi'(Y). \label{eq:diff}
\end{align}
Multiplying both sides of \eqref{eq:diff} by $1/\varphi'(1)^2$, we obtain
\begin{align*}
    \frac{\varphi'(AY)}{\varphi'(1)} = \frac{\varphi'(A)}{\varphi'(1)} \cdot \frac{\varphi'(Y)}{\varphi'(1)}. 
\end{align*}
This results in the following functional equation:
\begin{align}
    \tilde{\varphi}(AY) = \tilde{\varphi}(A)\tilde{\varphi}(Y), \label{eq:uay}
\end{align}
where the function $\tilde{\varphi}$ is defined as
\begin{align*}
    \tilde{\varphi}(Y) = \frac{\varphi'(Y)}{\varphi'(1)}.
\end{align*}
Under the assumption that $\tilde{\varphi}$ is continuous, the general solution to the functional equation \eqref{eq:uay} is denoted by $\tilde{\varphi}(Y)=Y^c$, where $c\in\mathbb{R}$ is an arbitrary constant \cite[1, pp.28--30]{Aczel}.
Therefore, we obtain the following differential equation:
\begin{align*}
    \frac{\varphi'(Y)}{\varphi'(1)} = Y^c.
\end{align*}
The general solution to this differential equation is expressed by
\begin{align*}
    \varphi(Y)=
    \begin{cases}
        \varphi'(1)\frac{Y^{1+c}-1}{1+c} + \varphi(1), &(c>-1), \\
        \varphi'(1)\log Y+\varphi(1), &(c=-1).
    \end{cases}
\end{align*}
From Definition \ref{def:FDPD}, it follows that $1+c\geq0$ and $\varphi'(1)>0$.
By setting $\zeta=1+c$, so that $\zeta\geq0$, $\varphi(1)=0$ and $\varphi'(1)=1$, we obtain
\begin{subnumcases}{\varphi(z)=}
    \frac{z^{\zeta}-1}{\zeta}, &$(\zeta> 0)$, \label{eq:FDPD_func_zeta}\\
    \log z, &$(\zeta=0)$. \label{eq:FDPD_func_zeta0}
\end{subnumcases}
Eqs. \eqref{eq:FDPD_func_zeta} and \eqref{eq:FDPD_func_zeta0} generate the JHHB divergence family \eqref{eq:JHHB_KL_zeta} and \eqref{eq:JHHB_KL}, respectively.
Based on $\bar{h}(A)=A^{-\zeta}$, the scale function of the affine transformation is denoted by $h(\bm{\Sigma}, \bm{\mu})=|\det\bm{\Sigma}|^{-\zeta}$ with $\zeta\geq0$ and $\gamma=0$.
\qed

\section{Proof of Theorem \ref{thm:hd_jones_div}}

Let $\tau$ be a strictly increasing function. 
Suppose that the JHHB score family and the H\"{o}lder score are equivalent.
Then, we have
\small
\begin{subnumcases}{-\tau\left(-\eta\left(\frac{\left\langle g f^\gamma \right\rangle}{\left\langle f^{1+\gamma} \right\rangle}\right) \left\langle f^{1+\gamma} \right\rangle\right)=}
    -\frac{1+\gamma}{\zeta}\left\langle g f^\gamma \right\rangle^\zeta + \frac{\gamma}{\zeta}\left\langle f^{1+\gamma} \right\rangle^\zeta + \frac{1}{\zeta}, &$(\zeta>0)$, \label{eq:JHHB_xi_H}\\
    -(1+\gamma)\log\left\langle g f^\gamma \right\rangle + \gamma\log\left\langle f^{1+\gamma} \right\rangle, &$(\zeta=0)$, \label{eq:JHHB_xi_H0}
\end{subnumcases}
\normalsize
where the JHHB score family is obtained by substituting $\varphi(z)=(z^\zeta-1)/\zeta$ for $\zeta>0$ and $\varphi(z)=\log z$ into the FDPS \eqref{eq:FDP_score_gamma}.
Setting $g = f$, we obtain
\begin{subnumcases}{\tau(z)=}
    \frac{z^\zeta-1}{\zeta}, &$(\zeta>0)$, \label{eq:xi_zeta}\\
    \log z, &$(\zeta=0)$. \label{eq:xi_zeta0}
\end{subnumcases}
In the case of $\zeta=0$, substituting \eqref{eq:xi_zeta0} into \eqref{eq:JHHB_xi_H0} yields $\eta(z)=-z^{1+\gamma}$, which is the lower bound of $\eta$ \cite{holder-div1}.
By substituting \eqref{eq:xi_zeta} into \eqref{eq:JHHB_xi_H}, we obtain the following equation:
\begin{align*}
    (1+\gamma)\left( \frac{X}{Y} \right)^\zeta - \gamma 
    = \left( -\eta\left(\frac{X}{Y}\right) \right)^\zeta,
\end{align*}
where we put $X=\langle gf^\gamma\rangle$ and $Y=\langle f^{1+\gamma}\rangle$.
Furthermore, let $z =X/Y$, we derive the function $\eta$ for the case $\zeta > 0$ as
\begin{align}
    \eta(z) = -{\rm sign}((1+\gamma)z^\zeta-\gamma)\cdot\left|(1+\gamma)z^\zeta-\gamma\right|^{\frac{1}{\zeta}}. \label{eq:JHHB_H_func}
\end{align}

We verify that function $ \eta $ satisfies the conditions $\eta(1)=-1$ and $\eta(z)\geq-z^{1+\gamma}$ for all $z\geq0$.  
Thus, it is easy to verify that $\eta(1) = -1$. 
Subsequently, we show that $ \eta(z) \geq -z^{1+\gamma} $.
From \eqref{eq:JHHB_H_func}, when $(1+\gamma) z^\zeta - \gamma \geq 0$, we have $-\mathrm{sign}((1+\gamma) z^\zeta - \gamma) = -1$, and hence
\begin{align*}
    (1+\gamma) z^\zeta - \gamma \leq z^{\zeta(1+\gamma)}.
\end{align*}
We put $t = z^\zeta$, this inequality becomes
\begin{align*}
    (1+\gamma)t - \gamma \leq t^{1+\gamma}.
\end{align*}
Because the right-hand side is a convex function for $\gamma>0$ and the left-hand side is the tangent line to $t^{1+\gamma}$ at $t = 1$, the inequality holds.
Conversely, from \eqref{eq:JHHB_H_func}, when $(1+\gamma) z^\zeta - \gamma < 0$, we have $-\mathrm{sign}((1+\gamma) z^\zeta - \gamma) = 1$, and thus
\begin{align*}
    \left[ (1+\gamma) z^\zeta - \gamma \right]^{\frac{1}{\zeta}} \geq -z^{1+\gamma}.
\end{align*}
Therefore, in both cases of \eqref{eq:JHHB_H_func}, we conclude that $\eta(z) \geq -z^{1+\gamma}$ for all $z\geq0$.
\qed

\section{Proof of Theorem \ref{thm:fdp_inq}}
By setting $\eta(z)=\gamma-(1+\gamma)z$, $\xi(z)=\varphi(z)$ in \eqref{eq:3HS}, and using the inequality $\eta(z)\geq-z^{1+\gamma}$, we obtain 
\begin{align}
    &\gamma \varphi(\langle f^{1+\gamma}\rangle) - (1+\gamma)\varphi(\langle gf^\gamma\rangle)\geq
    -\frac{\varphi(\langle gf^\gamma\rangle)^{1+\gamma}}{\varphi(\langle f^{1+\gamma}\rangle)^\gamma} \nonumber \\
    &=-\exp\left(-\left[\gamma\log\varphi(\langle f^{1+\gamma}\rangle)-(1+\gamma)\log \varphi(\langle gf^\gamma\rangle)\right]\right) \nonumber\\
    &= -\exp\left(-\left[\gamma \varphi_*(\langle f^{1+\gamma}\rangle) - (1+\gamma)\varphi_*(\langle gf^\gamma\rangle)\right]\right), \label{eq:fdp_lower}
\end{align}
where we define $\varphi_*(z)=\log \varphi(z)$.
If $\psi_*(z)=\varphi_*(e^z)=\log\varphi(e^z)$ is a strictly increasing and convex function, then the lower bound is an FDPS.
Thus, $\psi(z)=\varphi(e^z)$ being a strictly increasing and log-convex function is a necessary and sufficient condition for the lower bound in \eqref{eq:fdp_lower} to be an FDPS.
\qed

\section{Proof of Theorem \ref{thm:formula_V_and_U}}
\begin{lemma}\label{lemma:relation_V_and_U}
    Under Assumption \ref{assumption:holder_form}, $u(az+b)=v(z)$ and $zU(z)=aV(z)+b$ hold for all $z>0$, where $a\neq0$ and $b\in\mathbb{R}$ are an arbitrary constants.
\end{lemma}
\begin{proof}
Suppose that the nonnegative function $g\in\mathcal{F}$ is a constant function on $[0,k]$; that is,
\begin{align*}
    g(x)=\begin{cases}
        c, \quad& x\in[0,k],\\
        0, \quad& x\notin[0,k],
    \end{cases}
\end{align*}
where $c\geq0$ and $k>0$ are arbitrary constants.
Under Assumption \ref{assumption:holder_form}, the following equality
\begin{align}
    u(cU(c)) = v(V(c)) \label{eq:unit_int_func}
\end{align}
holds, where $u$ and $v$ are strictly increasing functions.
Thus, the inverse function of $u$ exists.
Eq. \eqref{eq:unit_int_func} can be rewritten as follows:
\begin{align}
    cU(c) = w(V(c)), \label{eq:cUTV}
\end{align}
where function $w$ is defined as
\begin{align}
    w(z)=u^{-1}(v(z)). \label{eq:w_inverse}
\end{align}
Subsequently, suppose that $g$ is a function assuming two distinct values on the unit interval
\begin{align*}
    g(x) = \begin{cases}
        c_1, &x\in[0, t],\\
        c_2, &x\in(t, 1],\\
        0, &x\notin[0,1],
    \end{cases}
\end{align*}
where $0\leq t\leq1$ and $c_1\geq0$ and $c_2\geq0$ are arbitrary constants.
Under Assumption \ref{assumption:holder_form}, we obtain the following equation,
\begin{align*}
    u(t c_1U(c_1)+(1-t)c_2U(c_2)) = v(t V(c_1)+(1-t)V(c_2))
\end{align*}
for all $c_1\geq0$ and $c_2\geq0$.
Using \eqref{eq:cUTV}, we obtain the following functional equation with respect to $w$:
\begin{align}
    t w(V(c_1)) +(1-t)w(V(c_2)) = w(t V(c_1) + (1-t)V(c_2)). \label{eq:Jensen_eq}
\end{align}
By Jensen's inequality, under the continuity of $w$, the general solution of \eqref{eq:Jensen_eq} must be both convex and concave; thus, it is denoted by $w(z)=az+b$, where $a\neq0$ and $b\in\mathbb{R}$.
Therefore, from \eqref{eq:cUTV} and \eqref{eq:w_inverse}, we have for all $z\geq0$,
\begin{align*}
    u(az+b)=v(z),
\end{align*}
and
\begin{align*}
    zU(z) = aV(z) + b.
\end{align*}
\qed
\end{proof}

\begin{lemma}[\cite{holder-div1}]\label{lemma:Kanamori}
    Under Assumptions 4.1 and 4.2 in \cite{holder-div1}, the functions $U$ and $V$ in Definition \ref{def:divergence} satisfy
    \begin{align*}
        V(z)=m\int zU'(z)dz, 
    \end{align*}
    for $z>0$, where $m\in\mathbb{R}\setminus\{0\}$ is a nonzero constant.
\end{lemma}

We prove Theorem \ref{thm:formula_V_and_U} using Lemmas \ref{lemma:relation_V_and_U} and \ref{lemma:Kanamori}.
Eq. \eqref{eq:H_assump} takes the form of a composite scoring rule (Definition \ref{def:divergence}), as it can be written by setting $T(x,y)=\eta(u(x)/v(y))v(y)$.
Therefore, Lemma \ref{lemma:Kanamori} is applicable.
The following equation holds by Lemmas \ref{lemma:relation_V_and_U} and \ref{lemma:Kanamori}:
\begin{align}
    \frac{1}{a}zU(z)-\frac{b}{a} = m\int zU'(z)dz, \label{eq:U_zU_dash}
\end{align}
where $b\in\mathbb{R}$ and $a, m\in\mathbb{R}\setminus\{0\}$ are constants.
Differentiating both sides of \eqref{eq:U_zU_dash} with respect to $z$, we obtain the following differential equation,
\begin{align}
    \frac{1}{a}(U(z)+zU'(z)) = mzU'(z). \label{eq:U_zU_dash2}
\end{align}
Assuming $am\neq1$, rearranging \eqref{eq:U_zU_dash2} yields the following differential equation,
\begin{align}
    \frac{U'(z)}{U(z)} = \frac{1}{am-1}\frac{1}{z}. \label{eq:U_zU_dash3}
\end{align}
Integrating both sides of \eqref{eq:U_zU_dash3} with respect to $z$ yields
\begin{align*}
    \log |U(z)| = \frac{1}{am-1}\log |z| + k_1,
\end{align*}
where $k_1\in\mathbb{R}$ is a constant.
Therefore, the general solution is expressed as
\begin{align*}
    U(z) = cz^{\frac{1}{am-1}},
\end{align*}
where $c=\pm e^{k_1}\in\mathbb{R}\setminus\{0\}$ is a constant.
From $U'(z)=\frac{c}{am-1}z^{\frac{1}{am-1}-1}$ and Lemma \ref{lemma:Kanamori}, the function $V$ is expressed as
\begin{align*}
    V(z) = m\int \frac{c}{am-1}z^{\frac{1}{am-1}}dz = \frac{c}{a}z^{\frac{am}{am-1}} + k_2,
\end{align*}
where $k_2\in\mathbb{R}$ is a constant.
From Lemma \ref{lemma:relation_V_and_U}, it must hold that $ak_2+b=0$.
From Assumption 4.2(b) in \cite{holder-div1} and Lemma \ref{lemma:relation_V_and_U}, it follows that $\lim_{z\to0}V(z)=0=V(0)$ and that $\lim_{z\to0}V'(z)$ must exist.
Therefore, it must hold that $1/(am-1)>0$ and $k_2=b=0$.
Thus, functions $U$ and $V$ are expressed by
\begin{align}
    U(z) &= cz^\gamma, \quad
    V(z) =\frac{c}{a}z^{1+\gamma}, \label{eq:func_UV}
\end{align}
where $\gamma>0$ and $a,c\in\mathbb{R}\setminus \{0\}$ are constants.
The functions $u$ and $v$ satisfy the following relation:
\begin{align}
    u(az)=v(z). \label{eq:relation_uv}
\end{align}
Substituting \eqref{eq:func_UV} and \eqref{eq:relation_uv} into \eqref{eq:H_assump} yields \eqref{eq:H_FDP_score}.
That is, it holds that
\begin{align*}
    S(g,f) = \eta\left(\frac{u(\langle gU(f)\rangle)}{v\left(\langle V(f)\rangle\right)}\right)v(\langle V(g)\rangle)=\eta\left(\frac{u(c\langle gf^\gamma\rangle)}{u(c\langle f^{1+\gamma}\rangle)}\right)u(c\langle f^{1+\gamma}\rangle).
\end{align*}
\qed

\section{Proof of Theorem \ref{thm:H_FDP_nonnegative}}
By calculating the lower bound of \eqref{eq:H_FDP_score}, we obtain
\begin{align*}
    S(g,f) &= \eta\left(\frac{u(c\langle gf^\gamma\rangle)}{u(c\langle f^{1+\gamma}\rangle)}\right)u(c\langle f^{1+\gamma}\rangle)
    \overset{(a)}{\geq}-\frac{u(c\langle gf^\gamma\rangle)^{1+\gamma}}{u(c\langle f^{1+\gamma}\rangle)^\gamma}\overset{(b)}{\geq}-u(c\langle g^{1+\gamma}\rangle),
\end{align*}
where inequality (a) follows from $\eta(z)\geq-z^{1+\gamma}$ for all $z\geq0$ and inequality (b) must hold because $S(g,f)\geq S(g,g)=-u(c\langle g^{1+\gamma}\rangle)$ for all $g,f\in\mathcal{F}_\gamma$.
By rewriting inequality (b), we obtain the following inequality:    
\begin{align*}
    \exp\left((1+\gamma)\log u(c\langle gf^\gamma\rangle) -\gamma\log u(c\langle f^{1+\gamma}\rangle)\right)\leq\exp\left(\log u(c\langle g^{1+\gamma}\rangle)\right).
\end{align*}
Taking the logarithm of both sides results in the following inequality:
\begin{align*}
    \log u(c\langle g^{1+\gamma}\rangle) - (1+\gamma)\log u(c\langle gf^\gamma\rangle) +\gamma\log u(c\langle f^{1+\gamma}\rangle) \geq 0.
\end{align*}
Let $\varphi(z)=\log u(cz)$.
Subsequently, we obtain the following inequality.
A constant multiple of the left-hand side of this inequality corresponds to the FDPD,
\begin{align*}
    \varphi(\langle g^{1+\gamma}\rangle) - (1+\gamma)\varphi(\langle gf^\gamma\rangle) +\gamma\varphi(\langle f^{1+\gamma}\rangle) \geq 0.
\end{align*}
From Corollary \ref{coro:FDPD}, the necessary and sufficient condition for the FDPD to be a divergence is that $\psi(z)=\varphi(e^z)$ is strictly increasing and convex.
Therefore, the necessary and sufficient condition for the composite scoring rule \eqref{eq:H_FDP_score} to be strictly proper is that $\log u(ce^z)$ is strictly increasing and convex, and $c$ is a positive constant.
Moreover, because $\varphi:[0,\infty)\to[-\infty,\infty]$, it follows that $u:[0,\infty)\to[0,\infty]$.
\qed

\end{document}